\documentclass[11  pt, oneside]{article}   	\usepackage[margin=1.25in,lmargin=0.75in,rmargin=1.75in,bmargin=0.75in,tmargin=0.75in]{geometry}      
\usepackage{graphicx}
\usepackage{booktabs} 
\usepackage{listings}
\usepackage{amsthm}
\newtheorem{theorem}{Theorem}

\newtheorem{lemma}[theorem]{Lemma}

\usepackage{amsmath}
\usepackage{amsfonts}
\usepackage{amssymb}
\usepackage{graphicx}
\usepackage{xspace}

\newcommand{\cut}{\textbf{cut}}

\DeclareMathOperator*{\minimize}{minimize}

\providecommand{\subjectto}{\ensuremath{\text{subject to}}}

\newcommand{\CC}{\ensuremath{\mathcal{C}}\xspace}
\newcommand{\kk}{\ensuremath{\mathcal{K}}\xspace}
\newcommand{\dd}{\ensuremath{\mathcal{D}}\xspace}
\newcommand{\lcc}{\textsc{LambdaCC}\xspace}

\newcommand{\cc}{\textsc{Correlation Clustering}\xspace}
\newcommand{\mcc}{\textsc{MotifCC}\xspace}
\newcommand{\mmcc}{\textsc{Mixed Motif Correlation Clustering}\xspace}
\newcommand{\muncut}{\textsc{Min Uncut}\xspace}
\newcommand{\maxcut}{\textsc{Max Cut}\xspace}

\newcommand{\mmccs}{\textsc{MMCC}\xspace}
\newcommand{\ccs}{\textsc{CC}\xspace}
\newcommand{\threeLP}{\textsc{ThreeLP}\xspace}
\newcommand{\cd}{\textsc{Cluster Deletion}\xspace}
\newcommand{\twocc}{\textsc{2-Correlation Clustering}\xspace}
\newcommand{\spc}{\textsc{Sparsest Cut}\xspace}
\newcommand{\NEPPC}{\ensuremath{\mathit{NEPPC}}}
\newcommand{\OPT}{\ensuremath{\mathit{OPT}}\xspace}
\newcommand{\ALG}{\ensuremath{\mathit{ALG}}\xspace}
\newcommand{\eps}{\ensuremath{\varepsilon}}
\newcommand{\dist}{\ensuremath{\mathit{dist}}}

\newcommand{\ee}{\ensuremath{\mathcal{E}}\xspace}
\newcommand{\LP}{\ensuremath{\mathrm{LP}}\xspace}

\usepackage{amsthm}
\usepackage{algorithm}
\usepackage[noend]{algpseudocode}
\usepackage{algorithmicx}

\usepackage{microtype}

\usepackage[T1]{fontenc}
\usepackage[utf8]{inputenc}
\usepackage{authblk}

\begin{document}
\title{Correlation Clustering Generalized}

\author[a]{David F. Gleich} 
\author[b]{Nate Veldt}
\author[c]{Anthony Wirth}

\affil[a]{Purdue University Computer Science Department}
\affil[b]{Purdue University Mathematics Department}
\affil[c]{The University of Melbourne, Computing and Information Systems School}
\date{}
\maketitle

\begin{abstract}
	We present new results for LambdaCC and MotifCC, two recently introduced variants of the well-studied correlation clustering problem.
	Both variants are motivated by applications to network analysis and community detection, and have non-trivial approximation algorithms.
	
	We first show that the standard linear programming relaxation of LambdaCC has a~$\Theta(\log n)$ integrality gap for a certain choice of the parameter~$\lambda$.
	This sheds light on previous challenges encountered in obtaining parameter-independent approximation results for LambdaCC.
	We generalize a previous constant-factor algorithm to provide the best results, from the LP-rounding approach, for an extended range of~$\lambda$.
	
	MotifCC generalizes correlation clustering to the hypergraph setting. In the case of hyperedges of degree~$3$ with weights satisfying probability constraints, we improve the best approximation factor
	from~$9$ to~$8$. We show that in general our algorithm gives a~$4(k-1)$ approximation when hyperedges have maximum degree~$k$ and probability weights.
	We additionally present approximation results for LambdaCC and MotifCC where we restrict to forming only two clusters. 
\end{abstract}

\section{Introduction}
\cc (\ccs), introduced by Bansal et al.~\cite{Bansal2004correlation}, is often viewed as a partitioning problem on signed graphs.
Given~$n$ nodes whose edges have so-called positive or negative weights (maybe both),
the goal is to find the clustering which \emph{correlates} as much as possible with the edge weights.
That is, a positive-weight edge suggests two nodes should be clustered together, while a negative-weight edge suggests separation, and these weights
are in some sense \emph{soft constraints}
There is a variety of settings for \cc, including different objective functions, and special classes of edge weights,
leading to a rich and interesting family of approximation algorithms and hardness results.

In this document, we consider two recent variants of the problem, called \textsc{Lambda Correlation Clustering} (\lcc)~\cite{veldt2018correlation}
and \textsc{Motif Correlation Clustering} (\mcc)~\cite{li2017motif}.
Although introduced independently, both problems are motivated by applications to community detection in unsigned graphs,
and are interesting to study from a theoretical perspective, each coming with non-trivial approximation guarantees.
\lcc is a generalization of the standard unweighted \ccs in which all positive edges have a common weight, while all negative
edges have another (possibly different) common weight.
A parameter~$\lambda$ determines these two weights and, implicitly, controls the size and structure of clusters formed by
optimizing the objective.
\mcc is a generalization of \cc to hypergraphs, designed to provide a framework for clustering graphs based on higher-order
subgraph patterns (i.e., \emph{motifs}). 
We present new results for \lcc and \mcc, not only where the number of clusters formed is
an outcome of minimizing the objective, but also where we (additionally) restrict to forming only two clusters.

\paragraph{Our results}
\begin{enumerate}
	\item We show that there exists some small~$\lambda$ such that the \lcc LP relaxation has a~$\Theta(\log n)$
	integrality gap.
	This hints at why constant-factor approximations have been developed for~$\lambda \geq 1/2$, but no analogous result has been found for \emph{small}~$\lambda$.
	We also extend the analysis of our previous algorithm for \lcc~\cite{veldt2018correlation} to outline the range of~$\lambda < 1/2$ values,
	that admit an approximation factor in~$o(\log n)$.
	\item We show that when we restrict to two clusters, \lcc reduces to the \muncut problem, which implies
	an $O(\sqrt{\log n})$ approximation for this special case~\cite{agarwal2005uncut}.
	\item We generalize the $4$-approximation of Charikar et al.\ for complete unweighted correlation clustering to obtain
	a~$4(k-1)$ approximation for \mcc on hypergraphs with edges of degree~$k$ where edge weights satisfy probability constraints.
	We consider the same LP relaxation as Li et al.~\cite{li2017motif}, and apply a similar rounding technique.
	However, we provide an approximation guarantee for arbitrary~$k$ that is linear in~$k$, in addition improving the
	factor for $k = 3$ from~$9$ to~$8$. 
	\item For \textsc{Two-Cluster MotifCC}, we design an algorithm that gives an asymptotic $1+k\,2^{k-2}$ approximation by generalizing the
	$3$-approximation of Bansal et al~\cite{Bansal2004correlation} for 2-\ccs (which applies when $k = 2$).
	This is the first combinatorial result for 2-\mcc, and is a~$7$-approximation for~$k=3$.
\end{enumerate}

\section{Background and Previous Results}
\label{sec:related}
In the most general formulation of \cc on (undirected) graphs -- excluding, for the moment, the generalization to
hypergraphs -- each pair of nodes~$(i,j)$  is assigned a pair of nonnegative weights $(w_{ij}^+, w_{ij}^-)$, i.e., a
similarity score and a dissimilarity score. In many cases, only one of these weights is assumed to be nonzero, to
indicate strict similarity or strict dissimilarity between pairs of nodes. We focus on the objective of \emph{minimizing disagreements}, which can be formally expressed as an integer linear program:
\begin{equation}
\label{eq:cc}
\begin{array}{lll} \text{minimize} & \sum_{i<j} w_{ij}^+ x_{ij} + w_{ij}^- (1-x_{ij})\\  \subjectto
& x_{ij} \leq x_{ik} + x_{jk} & \text{for all $i,j,k$} \\
& x_{ij} \in \{0,1\} & \text{for all $i < j$}
\end{array}
\end{equation}
The variable~$x_{ij}$ is~$1$ if nodes~$i$ and~$j$ are in separate clusters, and is~$0$ otherwise.
Thus, a clustering that separates $i,j$ incurs a penalty (also called a \emph{mistake}, or a \emph{disagreement}) of
weight $w_{ij}^+$, while if $i,j$ are together the penalty has weight~$w_{ij}^-$.
The objective of \emph{maximizing agreements} has also been extensively considered: it shares the same set of optimal
clusterings as minimizing disagreements, but is \emph{easier} from the perspective of approximations.
For the general weighted case, correlation clustering is equivalent to \textsc{Minimum
	Multicut}~\cite{demaine2006correlation}, which implies an $O(\log n)$ approximation, but also suggests that \cc (with
general weights) is unlikely to be approximated to within a constant factor in polytime~\cite{chawla2006hardness}. 
For weights satisfying probability constraints (i.e., $w_{ij}^+ + w_{ij}^- = 1$), Ailon et al.\ gave a $2.5$ approximation~\cite{ailon2008aggregating}. The best approximation factor for the standard unweighted problem (i.e., $(w_{ij}^+,w_{ij}^-) \in \{ (0,1), (1,0) \}$)
is slightly better than~$2.06$~\cite{chawla2015near}. 

\paragraph{Fixing the number of clusters}
In general, \cc does not require a user to specify number of clusters to be formed; the number of clusters arises
naturally by optimizing the objective. However, restricting the output of \cc to a fixed number of clusters has also been
studied extensively. In their seminal work, Bansal et al.\ showed a $3$-approximation for minimizing disagreements in the two-cluster
unweighted case (\twocc)~\cite{Bansal2004correlation}.
Later, Giotis and Guruswami showed a \emph{polynomial time approximation
	scheme} for maximizing agreements and for minimizing disagreements, when the number clusters is a fixed
constant~\cite{giotis2006correlation}. For the maximization version, \twocc is equivalent to \maxcut; based on this Dasgutpta et al.\ showed a 0.878-approximation for arbitrary weights~\cite{dasgupta2006algorithmic}. Extending Bansal et al.'s approach, Coleman et al.\ introduced faster, greedy $2$-approximations for minimizing disagreements for unweighted \twocc~\cite{coleman20082cc}, and gave a more extensive overview of the historical interest in this problem. Given this recurring interest in correlation clustering with a fixed number of clusters, we address several questions involving the two-cluster case in this manuscript.

\subsection{Lambda Correlation Clustering}
In previous work, we introduced the \lcc objective, which can be viewed as a special case of weighted correlation
clustering~\eqref{eq:cc} in which $(w_{ij}^+, w_{ij}^-) \in \{ (1-\lambda, 0), (0, \lambda) \}$ for some user-chosen
parameter $\lambda \in (0,1)$. This provides the following framework for partitioning unsigned networks: given an
unsigned graph $G = (V,E)$, treat each edge, in~$E$, as a \emph{positive edge} of weight~$(1-\lambda)$ in a signed
graph, and treat each non-edge as a \emph{negative edge} with weight~$\lambda$.
When $\lambda = 1/2$, \lcc amounts to unweighted \cc;
with \emph{small}~$\lambda$, \lcc amounts to \spc; and when~$\lambda$ is large, \lcc amounts to
\cd. {We previously outlined another, similar, edge-weighting scheme~\cite{veldt2018correlation}
	that is equivalent to the \textsc{Modularity} objective~\cite{newman2004modularity}.
	We do not consider it here, however, as this scheme does not appear to lead to new approximation results.}

For $\lambda > 1/2$, we gave a $3$-approximation based on the LP-rounding technique of van~Zuylen and
Williamson~\cite{zuylen2009deterministic}, and a $2$-approximation which holds specifically for $\lambda > |E| / (1+|E|)$,
hence, for \cd. We also note that when $\lambda > 1/2$, \lcc can be viewed as a specific case of the specially weighted correlation
clustering variant considered by Puleo and Milenkovic~\cite{puleo2015cc}, for which they gave a $5$-approximation based on a generalization of the LP rounding scheme of Charikar et al.~\cite{charikar2005clustering}. However, the proof strategies for all of these algorithms fail when considering arbitrarily small~$\lambda$.

\subsection{Motif Correlation Clustering}
Li et al.\ introduced a higher-order generalization of \cc, which they call \textsc{Motif Correlation
	Clustering} (\mcc), as a means for clustering networks based on higher-order motif patterns shared among
nodes~\cite{li2017motif}. This objective is motivated by previous successful results for motif-based graph clustering
(see e.g.,~\cite{benson2016}). Although a similar higher-order correlation clustering objective was considered by Kim et al.\ for image segmentation~\cite{kim2011highcc}, Li et al.\ were the first to study the objective from a theoretical perspective.
In their approach, we let~$E_k$ denote the set of all $k$-tuples of nodes in~$G$, and let each $\ee \in E_k$ have a positive weight, $w_\ee^+$, and a negative weight, $w_\ee^-$. If a clustering separates at least one pair of nodes in~\ee, this gives a penalty of $w_\ee^+$; otherwise, if all nodes in~\ee are clustered together, there is a penalty of $w_\ee^-$.
\mcc is formally expressed as the following ILP, a generalization of ILP~\eqref{eq:cc}:
\begin{align}
&\begin{array}{lll} 
\label{eq:ILP}
\minimize & \sum_{\ee \in E_k} w_\ee^+x_\ee + w_\ee^-(1-x_\ee)\\
\subjectto & x_{uv} \leq x_{uw} + x_{vw} &\text{ for all $u,v,w$} \\ &x_{uv} \in \{0,1\} & \text{
	for all $u<v$} \\ & x_{uv} \leq x_{\ee} & \text{ for all $u,v \in \ee$} \\& (k-1) x_\ee \leq \sum_{u,v \in \ee} x_{uv}
& \text{ for all $\ee \in E_k$} \\ &x_\ee \in \{0,1\} & \text{ for all $\ee \in E_k$}. \end{array} 
\end{align}
The first two constraints above ensure the variables encode a clustering ($x_{uv} = 1$ if $u,v$ are separated). Since~$x_\ee$
is binary, constraint $x_\ee \geq x_{uv}$ ensures that if any two nodes $u,v$ in $\ee$ are separated, then $x_\ee = 1$
(i.e., the $k$-tuple is split). The fourth constraint guarantees that $x_\ee = 0$ if all pairs of nodes in $\ee$ are together. Li et al.\ considered an even more general objective, which they referred to as \mmcc (\mmccs),
where motifs of multiple sizes are considered at once, and the objective is a positive linear combination of objectives
of the form~\eqref{eq:ILP} for different values of~$k$.
In their analysis they restrict to hyperedges of size~$2$ and~$3$, in other words they optimize an objective like this:
\[\minimize \hspace{.5cm}
{\textstyle \sum_{u<v}
	w_{uv}^+ x_{uv}+
	w_{uv}^- (1-x_{uv}) + \sum_{\ee \in E_3} w_\ee^+x_\ee + w_\ee^-(1-x_\ee)}\,.\]
For this setting, they show a $9$-approximation for the problem when hyperedge weights satisfy probability constraints
($w_\ee^+ + w_\ee^- = 1$, for every hyperedge~\ee of size~$2$ or~$3$). 
Recently, Fukunga gave an $O(k \log n)$ approximation for general weighted hypergraphs by rounding the same LP~\cite{fukunga2018highcc}.

\section{New Results for \lcc}
Given a signed graph,~$G$, in which every pair of nodes is part of a negative edge set,~$E^-$, or a positive edge set,~$E^+$, the linear program relaxation of \lcc is
\begin{equation}
\label{eq:lcc}
\begin{array}{lll} \text{minimize} & \sum_{(i,j) \in E^+} (1-\lambda) x_{ij} \,\, + & \sum_{(i,j) \in E^-} \lambda (1-x_{ij})\\  \subjectto
& x_{ij} \leq x_{ik} + x_{jk} & \text{for all $i,j,k$} \\
& 0 \leq x_{ij} \leq 1 & \text{for all $i < j$}
\end{array}
\end{equation}
Although a constant-factor approximation for \lcc exists for~$\lambda \geq 1/2$, by rounding LP~\eqref{eq:lcc}, we show that there exists
some small~$\lambda$ such that the integrality gap is~$O(\log n)$.
We then give parameter-dependent approximation guarantees for small~$\lambda$, and consider new results for two-cluster \lcc.

\subsection{Integrality Gap for the \lcc Linear Program}
\label{sec:logproof}
Demaine et al.\ prove that the integrality gap for the general weighted
\cc LP relaxation is $O(\log n)$~\cite{demaine2006correlation}. This does not immediately imply anything for our
specially weighted case, but adapting some of their ideas, and adding some non-trivial steps, does reveal an $O(\log n)$ integrality gap for the \lcc linear program relaxation. The proof takes the following steps.
\begin{enumerate}
	\item Construct an instance of \lcc from an expander graph,~$G$.
	\item Prove that, because of the expander properties of~$G$, the optimal \lcc clustering must make~$\Omega(n)$ mistakes.
	\item Demonstrate the LP relaxation has a feasible solution with a score of $O({n}/{\log n})$.
\end{enumerate}

In order to accomplish third step listed above, we do not (necessarily) produce a feasible solution for the standard LP relaxation of \lcc: in particular, in our solution triangle constraints are not guaranteed. Instead, we produce a feasible solution for a related linear program considered by Wirth in his PhD thesis~\cite{wirth2004approximation}. The fundamental construct of this LP is the \textsc{Negative Edge with Positive Path Cycle} (NEPPC), where, $\NEPPC(i_1,i_2, \hdots , i_m)$ represents a sequence (a \emph{path}) of (positive) edges, $(i_1,i_2), (i_2, i_3), \hdots ,(i_{m-1},i_m) \in
E$, with a single (negative) non-edge completing the cycle: $(i_1, i_m) \notin E$. 
For \lcc, defined on a graph $G = (V,E)$, with parameter~$\lambda \in (0,1)$, we have the linear program:
\begin{equation}
\label{neppc}
\begin{array}{lll} \text{minimize} & \sum_{(i,j)\in E} (1-\lambda)x_{ij} \,\,+ & \sum_{(i,j) \notin E} \lambda (1-x_{ij})\\ \subjectto
& x_{i_1, i_m}
\leq
\sum_{j=1}^{m-1} x_{i_j, i_{j+1}}
&\text{ for all $\NEPPC(i_1,i_2, \hdots , i_m)$} \\
& x_{ij}
\leq
1
& \text{ for all $(i,j)\notin E$} \\
& 0
\leq
x_{ij}
&\text{ for all $(i,j)$}\,.
\end{array}
\end{equation}
Wirth~\cite{wirth2004approximation} proved that the set of optimal solutions to the NEPPC linear program~\eqref{neppc}
is exactly the same as the optimal solution set to the \cc LP, the relaxation of ILP~\eqref{eq:cc}.\footnote{Although the proof is shown for the unweighted case, we note that all aspects of the proof immediately carry over to the weighted
	case.} Since a feasible solution for the \lcc NEPPC linear program~\eqref{neppc} is an upper bound on the optimum for~\eqref{neppc},
which is the same as the optimum for the standard \lcc LP, we can bound the optimum of the latter. We now prove our result:
\begin{theorem}
	There exists some $\lambda$ such that the integrality gap of LP~\eqref{eq:lcc} is $O(\log n)$.
\end{theorem}
\begin{proof}
	We separate the proof into the three steps outlined at the beginning of the section: constructing a \lcc instance from an expander graph, bounding the \lcc solution from below, and then upper bounding the LP relaxation.
	
	\paragraph{Constructing an instance of \lcc from an expander} 
	Let $G = (V,E)$ be a $(d,c)$-expander graph, where both $d$ and $c$ are constants (Reingold et al.\ proved that such expanders exist~\cite{reingold2000entropy}).
	That is,~$G$ is $d$-regular, and for every~$S \subset V$ with~$|S| \leq n/2$, we have
	\[\frac{\cut(S)}{|S|} \geq c  \implies  \frac{\cut(S)}{|S|} + \frac{\cut(S)}{|\bar{S}|} \geq c \implies \frac{\cut(S)}{|S| |\bar{S} | } \geq \frac{c}{n} \]
	where $\cut(S)$ denotes the number of edges between $S$ and $\bar{S} = V \backslash S$. Define the \emph{scaled sparsest cut} of a set $S$ to be ${\cut(S)}/{(|S| |\bar{S}|)}$ and let $\lambda^*$ minimize this ratio over all possible sets $S \subset V$. In previous work we showed that for any $\lambda \leq \lambda^*$, the optimal \lcc clustering places all nodes into one cluster, but there exists a range of $\lambda$ values slightly larger than $\lambda^*$ such that the optimum clustering coincides with a partitioning that produces the scaled sparsest cut score~\cite{veldt2018correlation}. For the expander graph we consider, this~$\lambda^*$ is at most the scaled sparsest cut score obtained by setting~$S$ to be a single node, so we have these upper and lower bounds on~$\lambda^*$:
	${c}/{n} \leq \lambda^* \leq {d}/(n-1)$.

	Let~$S^*$ be a set inducing an optimal scaled sparsest cut partition: $\lambda^* = {\cut(S^*)}/(|S^*| |\bar{S^*}|)$.
	From Theorem~3.2 in our previous work~\cite{veldt2018correlation}, we know that there exists some~$\lambda'$, slightly larger than~$\lambda^*$ whose
	optimum \lcc solution is the bipartition $\{S^*,\bar{S^*}\}$; let the \lcc score of this solution be~\OPT, and let~$\eps = \lambda'-\lambda^*$.
	We can choose~$\eps>0$ to be arbitrarily small, so it suffices to assume $\lambda' < 2\lambda^*$.
	
	\paragraph{Bounding $\OPT$ from below}
	With our choice of~$\lambda'$, by definition,
	\begin{align*}
	\OPT &= \cut(S^*) - \lambda' |S^*| |\bar{S}^*| + \lambda' \left( {n \choose 2} - |E| \right) \\
	&= \cut(S^*) - \lambda^* |S^*| |\bar{S}^*| - \varepsilon |S^*| |\bar{S}^*| + \lambda' \left( {n \choose 2} - |E| \right) \\
	&= 0 - \varepsilon |S^*| |\bar{S}^*| + \lambda^* \left( {n \choose 2} - |E| \right) + \varepsilon \left( {n \choose 2} - |E| \right)  \\
	&= \lambda^* \left( {n \choose 2} - |E| \right) +  \varepsilon \left( {n \choose 2} - |E| - |S^*| |\bar{S}^*|  \right)  \\
	 &= \lambda^* \left( \frac{n(n-1)}{2} - \frac{nd}{2} \right) +  \varepsilon \left( \frac{n(n-1)}{2} - \frac{nd}{2} - |S^*| |\bar{S}^*|    \right)  \\
	&\geq \lambda^* \left( \frac{n(n-1)}{2} - \frac{nd}{2} \right) +  \varepsilon \left( \frac{n(n-1)}{2} - \frac{nd}{2} - \frac{n^2}{4}  \right)  \\
	& \geq \frac{c}{n} \left( \frac{n(n-1)}{2} - \frac{nd}{2} \right) = \Omega(n)\,,
	\end{align*}
	relying on the definition of $\lambda^*$, the fact that $|E| = nd/2$ in this expander graph,
	and the bound $|S^*| |\bar{S}^*| \leq n^2/4$.
	
	\paragraph{Upper bounding the NEPPC LP}
	We now show that a carefully crafted feasible solution for the NEPPC LP~\eqref{neppc} has score $O(n/ \log n)$. Let~$\dist(i,j)$ denote the minimum path length between nodes~$i$ and~$j$ in~$G$, based on unit-weight edges~$E$.
	We are assuming the graph is connected, so each~$\dist(i,j)$ is a finite integer.
	(If the graph is not connected, we ought to solve \lcc on each connected component separately.)
	Consider the following setting of values~$x_{ij}$:
	\[ x_{ij} = \begin{cases}
	{2}/(\log_d n)
	& \text{if $(i,j) \in E$} \\
	1
	& \text{if $(i,j) \notin E$ and $\dist(i,j) \geq (\log_d n) /2$} \\
	0
	& \text{if $(i,j) \notin E$ and $\dist(i,j) <(\log_d n) /2$}\,.
	\end{cases}
	\]
	We show that this is feasible for the NEPPC LP~\eqref{neppc}.
	Since all (positive) edges are assigned the same LP score, the NEPPC constraints are satisfied at a (negative) non-edge,
	$(i,j)$, if and only if $x_{ij} \leq \dist(i,j) \cdot 2/(\log_d n)$.
	When $\dist(i,j)$ is less than $\log_d(n)/2$,  $x_{ij} = 0$, so this inequality is trivially true.
	When $\dist(i,j)$ is at least $\log_d(n)/2$,
	the NEPPC inequality is true because $\dist(i,j) \cdot 2/(\log_d n)$ is at least~$1$, which is~$x_{ij}$.
	
	For constant~$d$, the contribution from the (positive) edges  to LP~\eqref{neppc} is:
	\[ (1-\lambda')|E|{2}/({\log_d n}) = {(1-\lambda') (nd)}/({\log_d n}) = O(n/\log n)\,.\] 
	From the (negative) non-edges, since the factor is~$1-x_{ij}$, we only have a non-zero contribution from the set of $(i,j) \notin E$ such that
	$\dist(i,j) < (\log_d n)/2 = \log_d \sqrt{n}$.
	For each node $v \in V$, there are at most $d^{\log_d  \sqrt{n} } = \sqrt{n}$ nodes within this distance; the total
	number of non-edges that contribute to the LP cost is therefore in~$O(n \sqrt{n})$.
	Each has a weight~$\lambda' < 2\lambda^*$, so
	\[ \text{LP contribution of non-edges} \leq \lambda' n \sqrt{n} \leq  (2d/(n-1)) \, n \sqrt{n} = O(\sqrt{n}) \leq O(n/ \log n).\]
	Therefore, the total LP cost corresponding to this feasible solution to NEPPC LP~\eqref{neppc} is $O(n/ \log n)$.
	Since the optimal \lcc solution has cost~$\Omega(n)$, we have shown that there exists some $\lambda < 1/2$ such that the LP relaxation~\eqref{eq:lcc} has an integrality gap of $O(\log n)$.
\end{proof}

\subsection{Parameter-Dependent Approximation Guarantees}
We now describe improved approximation guarantees for ranges of~$\lambda$ below~$1/2$, extending the
analysis of our previous $3$-approximation for $\lambda \geq 1/2$~\cite{veldt2018correlation}.
This $3$-approximation is obtained by solving the LP relaxation, forming a new unweighted signed graph~$G'$,
and then applying the \emph{pivoting} procedure, which repeatedly selects a node and clusters it with its positive neighbors. The approximation guarantee comes from applying a theorem of van Zuylen and Williamson for deterministic pivoting algorithms for correlation clustering~\cite{zuylen2009deterministic}. We give a full proof of the following result in the Appendix:
\begin{theorem}
	\label{thm:lamcc}
	Let~$(x_{ij})$ be the variables from solving the \lcc LP relaxation,
	and form a new unweighted \cc input~$G'$ by putting a positive edge between~$i$ and~$j$, if $x_{ij} \leq 1/3$ and a negative edge otherwise.
	Applying a pivoting algorithm to~$G'$ yields a clustering that is a $3$-approximation for~$\lambda > 1/2$, and an
	$\alpha$-approximation otherwise, where $\alpha = \max  \{\,{1}/{\lambda}, ({6-3 \lambda})/({1 +\lambda})  \}$. 
\end{theorem}
This theorem implies an approximation better than~$4.5$ for all~$\lambda \in (0.2324,0.5)$, but shows that the algorithm
performs worse and worse as~$\lambda$ decreases. However, for all~$\lambda$ in $\omega(1/\log n)$, this outputs a
better result than the standard, $O(\log n)$, rounding scheme.

\subsection{Two-Cluster LambdaCC}
Before moving on we present a theorem which implies an approximation guarantee and a hardness result for the two-cluster variant of \lcc.
\begin{theorem}
	\label{thm:2lcc}
	Two-cluster \lcc can be reduced to the weighted \muncut problem. An instance of \muncut with non-zero optimum can be reduced to an instance of two-cluster \lcc whose objective score for any clustering differs by at most a small constant factor.
\end{theorem}
\begin{proof}
	When we restrict to forming only two clusters, it is known that there is a direct equivalence between the max-agree objective for general \cc
	(where edges are unweighted, but some pairs of nodes might have no edge) and \maxcut~\cite{coleman20082cc}.
	A \maxcut instance can be viewed simply an instance of 2-\ccs with only negative edges,
	and an instance of 2-\ccs can be converted into an instance of \maxcut by replacing each negative edge with a pair of positive edges meeting at a new node.
	We observe that, by the same reductions, the minimization version of 2-\ccs is equivalent to \muncut, implying an $O(\sqrt{\log n})$ approximation for this objective~\cite{agarwal2005uncut}.
	Since this approximation result holds for weighted \muncut as well (see e.g., \cite{kale2007efficient}), we know we can reduce two-cluster \lcc to weighted \muncut to obtain an~$O(\sqrt{\log n})$ approximation. This has important ramifications even without the restriction on the number of clusters; \lcc is guaranteed to form two clusters for a certain parameter regime near~$\lambda^*$~\cite[Theorem~3.2]{veldt2018correlation}.
	
	In fact, \muncut can be reduced to a special case of two-cluster \lcc in the following way.
	Given a graph $G = (V,E)$, on which we wish to perform \muncut, construct a signed graph $G' = (V,E^+, E^-)$ by setting $E^- = E$ and $E^+ = (V \times
	V)\backslash E$. Give all edges in $E^+$ weight $(1-\lambda)$ and edges in $E^-$ weight $\lambda$, where $\lambda = |E^+|/(1+|E^+|)$ is chosen so that
	$\lambda/(1-\lambda) = |E^+|$.
	Let~$\CC$ encode a two-clustering of the nodes, let~$P(\CC)$ be the number of positive edge mistakes made by~\CC, and~$N(\CC)$ denote the number of
	negative mistakes.
	Then the \lcc objective corresponding to the clustering is $\textsc{LamCC}(\mathcal{C}) = (1-\lambda)P(\mathcal{C}) + \lambda N(\mathcal{C})$ and the number of
	edges in~$G$ that~$\mathcal{C}$ does not cut is $\textsc{Uncut}(\mathcal{C}) = N(\mathcal{C})$.
	Assuming~$G$ is not bipartite (in which case we could solve \muncut on~$G$, as well as 2-\lcc on~$G'$, in polynomial time),
	we know that $N(\mathcal{C}) \geq 1$, which means $\textsc{LamCC}(\mathcal{C}) \geq \lambda $.
	By our choice of~$\lambda$, we also note that $(1-\lambda)P(\mathcal{C}) \leq (1-\lambda) |E^+| = \lambda \leq \lambda\, \textsc{Uncut}(\mathcal{C})$. Thus,
	\[ \lambda\,\textsc{Uncut}(\mathcal{C}) \leq \textsc{LamCC}(\mathcal{C}) =
	(1-\lambda)P(\mathcal{C}) + \lambda \textsc{Uncut}(\mathcal{C}) \leq 2\lambda \textsc{Uncut}(\mathcal{C})\,, \]
	so the \lcc objective on~$G'$ is within factor two of the \muncut objective on~$G$, after scaling it by a factor of~$\lambda$ (which, unless graph~$G$
	is almost complete, is close to~$1$).
	Since we know it is NP-hard to approximate \muncut to within a constant factor if the Unique Games conjecture is true~\cite{khot2002on,khot2005ugc}, the same hardness result holds in general for 2-\lcc.
\end{proof}

\section{Motif Correlation Clustering}
\label{sec:motifcc}
We now turn to improved approximations for \mcc.
We begin by presenting a~$4(k-1)$ approximation algorithm for the problem for hyperedges of degree~$k$ with edge weights satisfying probability constraints. We then consider a first step towards algorithms that do not rely on solving an
expensive LP relaxation, by showing how to obtain a combinatorial approximation for two-cluster \mcc (2-\mcc) for complete, unweighted instances.

\subsection{The $4(k-1)$ approximation}
Our algorithm for \mcc is closely related to the approach of Li et al.~\cite{li2017motif} and directly generalizes the
LP-rounding technique of Charikar et al.~\cite{charikar2005clustering}, which is itself an instantiation of the more general rounding procedure given in Algorithm~\ref{alg:cgw}.
The general algorithm forms clusters based on threshold parameters~$\gamma$ and~$\delta$, which are part of the input.
\begin{algorithm}[tb]
	\caption{Generalized CGW for Minimizing Hyper-Disagreements}
	\begin{algorithmic}[5]
		\State{\bfseries Input:} Signed hypergraph $G = (V,E_k)$, 
		and threshold parameters~$\gamma$ and~$\delta$
		\State Solve the LP-relaxation of ILP~\eqref{eq:ILP}, obtaining \emph{distances}~$(x_{ij})$
		\State $W \gets V$, $\mathcal{C} \gets \varnothing$
		\While{$W\neq \varnothing$ }
		\State Choose $u \in W$ arbitrarily, and define $T_u \gets \{i \in W\backslash\{u\}: x_{ui} \leq \gamma \}$
		\If{$\sum_{i \in T_u} x_{ui} < \gamma\delta |T_u|$}
		\,$S := \{u\} \cup T_u$
		\Else {\,\, $S := \{u\}$}
		\EndIf 
		\State $\mathcal{C} \gets \mathcal{C} \cup \{S \}$, $W \gets W\backslash S $
		\EndWhile
	\end{algorithmic}
	\label{alg:cgw}
\end{algorithm}
Charikar et al.\ proved that for the $k=2$ unweighted case of \mcc,
setting $\gamma = \delta = 1/2$ leads to a $4$-approximation.
Li et al.\ generalized this to obtain a $9$-approximation for $k = 3$ in the more general probability constrained case, by selecting $\gamma = \delta = 1/3$~\cite{li2017motif}.
Although they did not provide an analysis for motifs of size $k > 3$, it appears that their strategy of setting $\gamma = \delta = 1/k$ would at best lead to a~$k^2$ approximation. In contrast, we analyze a choice of parameters which leads to an approximation that is
\emph{linear} in~$k$.

The result is somewhat detailed, and we begin with some notation.
Let the family of~$k$-tuples be~$E_k$, and let $W \subseteq V$ be the subset of nodes in $G$ that remain unclustered after a certain number of rounds of Algorithm~\ref{alg:cgw}.
When considering a vertex~$u \in W$ and a specific $k$-tuple $\ee$, it will be convenient to define~$a$ to be the node in~$\ee$ \emph{closest} to~$u$,
i.e., $\arg \min_{i \in \ee} x_{ui}$, while~$z$ is the \emph{farthest}, $\arg \max_{i \in \ee} x_{ui}$.
We have~$T_u$ similar to Algorithm~\ref{alg:cgw}, with $\gamma = 1/(2(k-1))$,
while~$T_u^k$ are those $k$-tuples that include~$u$, with all non-$u$ nodes in~$T_u$:
\begin{equation}
\label{eq:tuk}
T_u =  {\textstyle \left \{i \in W\backslash \{u\}: x_{ui} \leq \frac{1}{2(k-1)} \right\}} \quad \text{and} \quad T_u^k = \{
\ee \in E_k : u \in \ee \text{ and }  (\ee - \{u\}) \subset T_u \}\,.
\end{equation}
For~$z \notin T_u$, we let~$P_z$ be those $k$-tuples in which~$z$ is the farthest element from~$u$ and some~$a \in T_u$ is
closest, viz.                                                                   
\begin{equation}
\label{eq:pz}
P_z = \{  (a, j_2, j_3, \ldots ,j_{k-1}, z) \in E_k :  a\in T, x_{ua} \leq x_{u,j_2} \leq x_{u,j_3} \leq
\cdots \leq x_{uz} \}\,. 
\end{equation}
Finally,~$\LP(A)$ denotes the LP score associated with a subset~$A$ of the set of degree-$k$ hyperedges: $A \subseteq E_k$.
\begin{theorem} 
	\label{thm:motifthm}
	For constant~$k$, let $G = (V,E_k)$ be a hypergraph in which for all $\ee \in E_k$ the weights satisfy probability constraints, $w_\ee^+ + w_\ee^- = 1$. Applying Algorithm~\ref{alg:cgw} with $\gamma = 1/(2(k-1))$ and $\delta = 1/2$ outputs a clustering that is a $4(k-1)$-approximation to \mcc.
\end{theorem}
We start with a proof outline, establish three lemmata, and then give full details in Section~\ref{sec:motifproof}. At each step the algorithm forms a cluster $S_u$ around an arbitrary $u \in W$. This cluster is associated with a set of hyperedges $A_u$ that have either been cut or placed inside of $S_u$. If for each $S_u$ individually we can show that mistakes made at $A_u$ are within a fixed factor of the lower bound $LP(A_u)$, this will imply an overall bound for the entire clustering.

In forming a cluster around $u$, the algorithm first identifies a set of nodes $T_u$ whose LP distance to $u$ is at most a preliminary threshold $\gamma = 1/(2(k-1))$. To verify if $\{u\}\cup T_u$ will make a good cluster, the algorithm checks whether \emph{on average} the distance from $u$ to $T_u$ is below a tighter threshold $\gamma\delta = 1/(4(k-1))$. If this doesn't hold, we let $\{u\}$ remain a singleton cluster. In forming clusters, we only explicitly consider distance variables $x_{ij}$ for $(i,j) \in V \times V$. However, the \mcc objective and its LP relaxation both depend on the hyperedge variables $x_\ee$ for $\ee \in E_k$. Therefore, in order to bound the weight of hyperedge mistakes we must leverage the LP constraints to understand the relationships between distance and hyperedge variables. Lemma~\ref{lem:bounds} establishes several useful relationships we will need later. Also, because our algorithm makes decisions based on the average distance between $u$ and $T_u$, we must interpret what this means for the average value of hyperedge variables $x_\ee$ in certain sets of hyperedges that we are trying to account for (e.g. $P_z$ and $T_u^k$ in~\eqref{eq:tuk} and~\eqref{eq:pz}). Lemmata~\ref{lem:singleton} and~\ref{lem:hardest} address this task. In the following, we adopt the convention that $x_{ii} = 0$ for every node~$i \in V$.
\begin{lemma}
	\label{lem:bounds} 
	For all $\ee \in E_k$ and any $u \in V$,
	\begin{enumerate}
		\item $x_\ee \leq \sum_{i \in \ee} x_{ui}$,
		\item $x_\ee \leq x_{ua} + (k-1)x_{uz}$, and
		\item $x_\ee \geq x_{uz} - x_{ua}$.
	\end{enumerate}
\end{lemma}
\begin{proof}
	By the triangle inequality,
	$$\sum_{(i,j)} x_{ij} \leq \sum_{(i,j)} (x_{ui} + x_{uj}) = \sum_{i \in \ee} (k-1) x_{ui}\,.$$
	The fourth constraint in the LP relaxation of~\eqref{eq:ILP} states that~$(k-1) x_\ee \leq \sum_{(i,j)} x_{ij}$, so
	we can prove the first inequality in the Lemma.
	
	The second inequality in the Lemma follows from the first inequality and the definitions of~$a$ and~$z$.
	The third inequality arises from the first and third constraints in the LP relaxation: $x_\ee \geq x_{az} \geq x_{uz} - x_{ua}$.
\end{proof}

\begin{lemma}
	\label{lem:singleton} 
	For all~$u \in W \subseteq V$, 
	if $\,\sum_{i \in T_u} x_{ui} \geq \beta |T_u|$, then $\sum_{\ee \in T_u^k} x_{\ee} \geq \beta |T_u^k|$.
\end{lemma}

\begin{proof}
	The set~$T_u$ comprises nodes that are \emph{close} to~$u$, in LP distance, while $T_u^k$ is the set of all $k$-tuples consisting of node~$u$ plus
	$k-1$ nodes taken from the set~$T_u$.
	For a fixed $i \in T_u$, let $K_i$ be the set of $(k-2)$-tuples of nodes in~$T_u$ that exclude~$i$.
	Then for all $\mathcal{K} \in K_i$, $\ee_{u,i,\mathcal{K}} = \{u,i, \mathcal{K}\}$ is a $k$-tuple containing~$u$ and~$i$,
	with corresponding variable $x_{ \ee\{u,i, \mathcal{K}\}} \geq x_{ui}$.
	Note that $|K_i| = {|T_u|-1 \choose k-2}$, and if we iterate through each $i \in T_u$ and count up $k$-tuples of the form $\ee_{u,i,\mathcal{K}}$ for $\mathcal{K} \in K_i$, we will count each $k$-tuple in $T_u \cup \{u\}$ exactly $k-1$ times. Thus:
	\begin{align*}
	\sum_{\ee \in T_u^k } x_{\ee} 
	&= \frac{1}{k-1} \sum_{i \in T_u}\sum_{\mathcal{K} \in K_i} x_{ \ee\{u,i, \mathcal{K}\}} \geq \frac{1}{k-1}\sum_{i \in T_u}\sum_{\mathcal{K} \in K_i}  x_{ui} \\
	&= \frac{1}{k-1} \sum_{i \in T_u} { |T_u| - 1  \choose k-2} x_{ui} \geq \frac{1}{k-1}{ |T_u| - 1  \choose k-2} \beta |T_u| =  \beta {|T_u|
		\choose (k-1)} = \beta |T_u^k|\,.
	\end{align*}
\end{proof}

\begin{lemma}
	\label{lem:hardest} 
	For all $\ee \in P_z$, let~$a_\ee$ denote the node in~$\ee$ closest to~$u$.
	If $\sum_{i \in T_u}\, x_{ui} < \beta |T_u|$,
	then $\sum_{\ee \in P_z} x_{ua_\ee} < \beta |P_z|$.
\end{lemma}

\begin{proof}
	We partition~$P_z$ into different sets, based on how many nodes are inside~$T_u$ and how many are outside~$T_u$,
	and then prove the inequality holds separately for each individual set.
	Define~$P_z^d$ to be those tuples in~$P_z$ in which~$a,j_2,\ldots,j_d \in T_u$, but $j_{d+1},\ldots \notin T_u$, 
	for $d = 1, 2, \hdots, k-1$. Let $J_u^d = \{ (a, j_2, \ldots, j_d) \subset T_u\}$ denote these $d$-tuples of nodes inside $T_u$, so $|J_u^d| = { |T_u| \choose d}$. For any $d$-tuple $\dd =  (a, j_2, \hdots, j_d) \in J_u^d$, define $x_\dd = \frac{1}{d}( x_{ua} + x_{u,j_2} + \cdots + x_{u,j_d})$ and note that $x_{ua} \leq x_\dd$ since $a$ is the node in $\dd$ closest to $u$. Observe that any node $i \in T_u$ shows up in exactly ${|T_u| - 1 \choose d-1}$ of the $d$-tuples in $J_u^d$. Therefore:
	\begin{align*}
	\sum_{(a, \hdots, j_d) \in J_u^d} x_{ua} &\leq \sum_{(a, \hdots, j_d) \in J_u^d} x_{\dd} = \sum_{(a, \hdots, j_d) \in J_u^d} \frac{1}{d} (x_{ua} + \cdots + x_{u,j_d}) \\
	& = \frac{1}{d} \sum_{i \in T_u} {|T_u| - 1 \choose d-1} x_{ui}  < \frac{1}{d} {|T_u| - 1 \choose d-1} |T_u| \beta = {|T_u| \choose d}\beta = |J_u^d| \beta.
	\end{align*}
	
	For every set of $k-d-1$ nodes $(j_{d+1}, j_{d+2}, \hdots, j_{k-1})$ outside of~$T_u$, satisfying $x_{u,j_{d+1}} \leq x_{u,j_{d+2}} \leq \cdots \leq x_{uz}$, the $k$-tuple $(\dd, j_{d+1}, j_{d+2}, \hdots , z)$ is in $P_z^d$ for every $\dd \in J_u^d$.
	We are now ready to perform a sum over tuples in~$P_z^d$, to show the desired result:
	\[ \sum_{(a,j_2, \hdots , j_{k-1},z) \in P_z^d} x_{ua} = \sum_{ \substack{(j_{d+1},\hdots , j_{k-1},z):\\ \gamma < x_{u,j_{d+1}} \leq \cdots
			\leq x_{uz}}} \sum_{\dd \in J_u^d} x_{ua} <
	\sum_{ \substack{(j_{d+1},\hdots , j_{k-1}):\\ \gamma < x_{u,j_{d+1}} \leq \cdots \leq x_{uz}}} |J_u^d| \beta = |P_z^d| \beta\,,  \]
	where~$\gamma$ is the threshold defining~$T_u$.
	Since the desired inequality holds for each~$P_z^d$ and $P_z = \bigcup_{d = 1}^{k-1} P_z^d$, the full result follows.
\end{proof}

\subsection{Proof of Theorem~\ref{thm:motifthm}}
\label{sec:motifproof}
\begin{proof}
	We must account for the weight of positive mistakes made at singleton clusters,~$\{u\}$,
	and the weight of both positive and negative mistakes made at non-singleton clusters.

	\paragraph{Singleton Clusters}
	Consider a cluster $S = \{u\}$.
	The algorithm incurs a penalty~$w_{\ee}^+$ for each~$\ee$ such that~$u \in \ee$.
	If some node $j \in \ee - \{u\}$ is not in~$T_u$, then the contribution to the LP score is
	$w_{\ee}^+ x_\ee$, which is at least $w_{\ee}^+ x_{uj}$, and therefore exceeds~$w_{\ee}^+ /(2(k-1))$.
	Thus the cost of the mistake at most~$2(k-1)$ times the LP penalty.
	
	It remains to account for all positive hyperedges in~$T_u^k$.
	Even if $w_\ee^+ = 1$ for all $\ee \in T_u^k$,
	$|T_u^k| = \binom{|T|}{k-1}$ is an upper bound on the total weight of mistakes made on hyperedges in~$T_u^k$.
	By the first observation of Lemma~\ref{lem:bounds}, and because~$u \in \ee$,
	\[ {\textstyle x_\ee \leq \sum_{i \in \ee} x_{ui} \leq (k-1) \frac{1}{2(k-1)} = \frac{1}{2}\,, \quad \text{hence,} \quad
		(1- x_\ee) \geq x_\ee}.\, \]
	Since~$w_\ee^+ + w_\ee^- = 1$, we can lower bound the contribution of~$T_u^k$ to the LP score:
	\begin{align*}
	\LP(T_u^k) &= \sum_{\ee \in T_u^k} w_\ee^+ x_\ee + w_\ee^- (1-x_\ee)
	\geq \sum_{\ee \in T_u^k} w_\ee^+ x_\ee + w_\ee^- x_\ee =\sum_{\ee \in T_u^k} x_\ee 
	\geq |T_u^k| \frac{1}{4(k-1)}\,,
	\end{align*}
	by Lemma~\ref{lem:singleton}, so we have paid for the mistakes within a factor~$4(k-1)$.
	
	\paragraph{Negative Mistakes at Non-Singletons}
	Next, we account for negative mistakes in clusters of the form $S = \{u\}\cup T$.
	Charikar et al.~showed that, when~$k =2$, these are accounted for within a factor~$4$; we prove the same for all~$k \geq 3$.
	For each~$\ee \in E_k$ such that~$\ee \subset S$,
	the algorithm makes a mistake of weight~$w_\ee^-$.
	On the other hand, the LP pays $w_\ee^- (1-x_\ee)$. Applying the first observation in Lemma~\ref{lem:bounds},
	\[ {\textstyle x_\ee \leq \sum_{i \in \ee} x_{ui} \leq k \frac{1}{2(k-1)} \leq \frac{3}{4}\,, \quad \text{hence,} \quad
		w_\ee^- (1-x_\ee) \geq \frac{w_\ee^- }{4}\,,} \]
	and we have the desired result for~$k \geq 3$.
	\paragraph{Positive Mistakes at Non-Singletons}
	A hyperedge~$\ee$ contained entirely within~$S = \{u\} \cup T$ incurs no positive-weight error.
	So, finally, we account for positive mistakes at hyperedges~$\ee$ where at least one node of~$\ee$ is in $S$
	and at least one node in~$\ee$ is $\notin S$.
	For each such hyperedge, we explicitly label the nodes of~$\ee$ with indices
	$a = j_1 < j_2 < \cdots < j_k = z$, with $x_{ua} = x_{u,j_1} \leq x_{u,j_2} \leq \hdots \leq x_{u,j_k} = x_{uz}$ where $a \in T_u$ and $z \notin T_u$.
	By the second and third observation in Lemma~\ref{lem:bounds} we know that
	\begin{equation}
	\label{eq:ineq}
	x_{uz} - x_{ua} \leq x_\ee \leq x_{ua} + (k-1) x_{uz}\,,
	\end{equation}
	
	First, if $a = u$, then we know $w_\ee^+ x_\ee \geq w_\ee^+(x_{uz} - x_{uu}) > w_\ee^+/(2(k-1))$,
	and we have individually accounted for each such positive mistake within a factor~$2(k-1)$. 
	If $a \neq u$ and $x_{uz} \geq 3/(4(k-1))$, we bound the mistake within factor~$4(k-1)$:
	\[ w_\ee^+ x_\ee \geq w_\ee^+(x_{uz} - x_{ua}) \geq w_\ee^+(3/(4(k-1)) - 1/(2(k-1)) = w_\ee^+/(4(k-1))\,.\]
	
	Finally, if~$a \neq u$ and  $x_{uz} \in \left( \frac{1}{2(k-1)}, \frac{3}{4(k-1)}\right)$,
	we account for all positive weights associated with edges in the following set, together:
	\[P_z = \{ \ee \in E_k :  \ee = (a, j_2, \ldots , z), a\in T, x_{ua} \leq x_{u,j_2} \leq x_{u,j_3} \leq \cdots
	\leq x_{uz} \}\,. \]
	The weight of mistakes made by the algorithm is $W_z^+ = \sum_{p \in P_z} w_p^+$,
	and we also define $W_z^- = \sum_{p \in P_z} w_p^-$.
	We start by observing that, since~$x_{ua} \leq x_\ee$ and~$W_z^+ + W_z^- = |P_z|$, due to probability constraints on
	weights, Lemma~\ref{lem:hardest} tells us that~$\sum_{\ee \in P_z} x_{ua} < (W_z^+ + W_z^-)/(4(k-1))$.
	\begin{align*}
	LP(P_z) &= \sum_{\ee \in P_z} w_\ee^+ x_\ee + w_\ee^- (1-x_\ee) \\  								
	&\geq \sum_{\ee \in P_z} w_\ee^+ (x_{uz} - x_{ua}) + w_\ee^- (1 - x_{ua} - (k-1) x_{uz}) \hspace{.5cm}\text{(by inequalities in~\eqref{eq:ineq})} \\		
	&= \sum_{\ee \in P_z} w_\ee^+ x_{uz} + w_\ee^- (1 - (k-1) x_{uz}) -  \sum_{\ee \in P_z} x_{ua} \\ 
	&\geq 	{\textstyle W_z^+ x_{uz} + W_z^- (1- (k-1)x_{uz}) - \frac{W_z^+ + W_z^-}{4(k-1)} } \hspace{.5cm} \text{(by the starting
		observation)}\\		
	&\geq 	{\textstyle W_z^+ \left( \frac{1}{2(k-1)} - \frac{1}{4(k-1)} \right) + W_z^- \left( 1 - \frac{1}{4(k-1)} - (k-1) \frac{3}{4(k-1)} \right) } \geq 	{\textstyle W_z^+ \frac{1}{4(k-1)} }\,,
	\end{align*}
	so the mistakes on all hyperedges in~$P_z$ are, collectively, accounted for within factor~$1/(4(k-1))$,
	concluding the Proof of Theorem~\ref{thm:motifthm}.
\end{proof}

We note that the approximation analysis given by Theorem~\ref{thm:motifthm} immediately extends to other variations of \mcc.
\paragraph{Extension~1}
Our analysis directly carries over to the \mmcc objective~\cite{li2017motif}, which includes penalties for all hyperedges up to size~$k$:
\[\min \,\, \sum_{t = 2}^k \rho_t \sum_{\ee \in E_t} w_\ee^+ x_\ee + w_\ee^- (1-x_\ee)\,,\]
where $\rho_t >0$ is a weight indicating how much we \emph{care} about motifs of size~$t$. Theorem~\ref{thm:motifthm} specifically considers
the case where $\rho_k = 1$ and $\rho_t = 0$ for $t < k$, but the analysis still holds for other combinations of~$\rho$-weights.
As noted by Li et al., it is sufficient to account for mistakes at the largest-sized motif~\cite{li2017motif}. We do note, however, that we will need to include $O(n^t)$ variables and constraints in the LP for each motif size $t$ where $\rho_t \neq 0$.
\paragraph{Extension~2}
We can consider a hybrid of the \lcc and \mcc objectives in which each degree-$k$ hyperedge is either positive with weight~$1-\lambda$ or negative with weight~$\lambda$.
If $\lambda \geq 1/2$, the proof of Theorem~\ref{thm:motifthm} still holds. Indeed, negative hyperedges are accounted for on an individual basis,
so weighting them more heavily has no effect.
When accounting for positive mistakes, we simply require an occasional extra line of algebra in which we note that for each negative hyperedge~$\ee$,
$\lambda (1- x_\ee) \geq (1-\lambda) (1-x_\ee)$. After applying this inequality, the result will follow through. 

As an example, after forming a singleton cluster~$\{u\}$, consider how to account for positive hyperedges that include~$u$ plus nodes inside the set~$T_u$.
We are no longer considering probability-constrained edges, so let~$E_k^+$ denote positive hyperedges (which all have weight~$1-\lambda$) and let~$E_k^-$ denote the set of negative hyperedges (with weight~$\lambda$).
The weight of mistakes made by the algorithm is at most $(1-\lambda)|T_u^k|$, which is the case if all $k$-tuples in~$T_u^k$ are positive. Then
\begin{align*}
\LP(T_u^k) &= \sum_{\substack{\ee \in T_u^k\\ \ee \in E_k^+ }} (1-\lambda) x_\ee + \sum_{\substack{\ee \in T_u^k\\ \ee \in E_k^-}} \lambda (1-x_\ee) \\
&\geq \sum_{\substack{\ee \in T_u^k\\ \ee \in E_k^+ }} (1-\lambda) x_\ee + \sum_{\substack{\ee \in T_u^k\\ \ee \in E_k^- }} (1-\lambda) (1-x_\ee) \\
&\geq \sum_{\substack{\ee \in T_u^k\\ \ee \in E_k^+ }} (1-\lambda) x_\ee + \sum_{\substack{\ee \in T_u^k\\ \ee \in E_k^- }} (1-\lambda) x_\ee \\
&= (1-\lambda) \sum_{\ee \in T_u^k} x_\ee \geq (1-\lambda)\frac{|T_u^k| }{4(k-1)}\,,
\end{align*}
by Lemma~\ref{lem:singleton}, so the mistakes are accounted for within the desired factor.


\subsection{Two-Cluster MotifCC}
The LP relaxation of \mcc involves~$O(n^k)$ variables and~$O(n^k)$ constraints for all~$k > 2$, and is therefore very expensive to solve in practice.
For standard \cc, only a few of the known approximation algorithms avoid solving an expensive convex
relaxation~\cite{ailon2008aggregating,Bansal2004correlation}; it is natural to ask whether a similar, combinatorial, approach
can be taken for \mcc.
We give first steps in this direction, with a constant-factor combinatorial approximation algorithm for \mcc,
when the output is restricted to two clusters,
generalizing the $3$-approximation of Bansal et al.\ for \twocc~\cite{Bansal2004correlation}. 
Our method is shown in Algorithm~\ref{alg:papt}. We call this algorithm \emph{Pick-a-Pivot-Tuple}, and show it satisfies the following result:
\begin{algorithm}[tb]                                                           
	\caption{\textsc{Pick-A-Pivot-Tuple}}                                              
	\label{alg:papt}
	\begin{algorithmic}[5]                                                  
		\State{\bfseries Input:} An instance of 2-\mcc:
		$G = (V,E_k)$ be a hypergraph where $(w_\ee^+,w_\ee^-) \in \{ (0,1), (1,0) \}$ for every $k$-tuple.
		\For{$(k-1)$-tuple $\kk \subseteq V$}
		\State $\mathcal{C}_{\kk} \gets$ the clustering formed by placing~\kk in a cluster with all~$u$ such that $\ee = \kk \cup \{u\}$ is positive, and placing all remaining nodes in the other cluster.
		\EndFor
		\State \textbf{Return} the $\mathcal{C}_{\kk}$ with fewest mistakes.
	\end{algorithmic}
\end{algorithm}
\begin{theorem}
	\label{thm:2motifcc}
	For a constant integer~$k > 1$, Algorithm~\ref{alg:papt} returns a $(1 + kc)$-approximation for 2-\mcc,
	where $c \leq 2^{k-2}$ for $k = 2,3$, while $\lim_{n \rightarrow \infty}c = 2^{k-2}$ for $k > 3$.
\end{theorem}

\begin{proof}
	
	Let \OPT be the minimum number of mistakes, and let \ALG be the number of mistakes made by the algorithm.
	In order to bound \ALG, assume for now that we have separated the nodes into an optimal partition,~$\mathcal{C}$.
	Consider what happens if we iterate through all $(k-1)$-tuples of nodes \kk that are unbroken in~$\mathcal{C}$. For each~$\kk$, form the clustering
	$\mathcal{C}_{\kk}$ by \emph{pivoting} on~$\kk$ and compare it against~$\mathcal{C}$.
	Let~$d_{\kk}$ be the number of nodes that~$\mathcal{C}_{\kk}$ moves to the \emph{wrong} side of the cut, compared to~$\mathcal{C}$.
	Each moved node~$w$ corresponds to a $k$-tuple $\kk\cup \{w\}$ at which~$\mathcal{C}$ makes a mistake.
	
	\paragraph{Bounding \OPT below} Let~$d$ be the minimum number of nodes that switch sides when we pivot around some $(k-1)$-tuple of nodes on the same side
	of~$\mathcal{C}$. For each such $(k-1)$-tuple~$\kk$, the optimal partition makes at least~$d$ mistakes at hyperedges containing all nodes in~$\kk$.
	Since there are~$k$ distinct ways to select $(k-1)$-tuple of nodes from a set of~$k$ nodes, $\OPT \geq \frac{dP}{k}$ where~$P$ is the total number of $(k-1)$-tuples that are located on the same side of~$\mathcal{C}$. Though we do not know~$P$ a priori, at minimum this is equal to $2 {n/2 \choose k-1}$, which is the case when $\mathcal{C}$ partitions the graph into two equally-sized clusters ($P$ would be larger if one side contained more than half the nodes). Thus, $OPT \geq 2d {n/2 \choose k-1}/k$.
	
	\paragraph{Bounding \ALG above} The clustering returned by our algorithm will be at least as good as $\mathcal{C}_{\kk^*}$, where~$\kk^*$ is the $(k-1)$-tuple
	of pivot nodes on the same side of~$\mathcal{C}$ that moves only~$d$ nodes from one side of~$\mathcal{C}$ to the other.
	Moving these~$d$ nodes contributes at most an extra $d {n-1 \choose k-1}$ mistakes in addition to mistakes that $\mathcal{C}$ already made. Therefore $ALG \leq OPT + d {n-1 \choose k-1}$. Using the observation that
	\[ \lim_{n \rightarrow \infty} \frac{{n-1 \choose k-1}}{{n/2 \choose k-1}} = \lim_{n\rightarrow \infty} \frac{\frac{(n-1)(n-2) \cdots
			(n-k+1)}{(k-1)!} }{ \frac{ \frac n2 (\frac n2 -1) \cdots (\frac n2 -k+2)}{(k-1)!} } = \lim_{n\rightarrow \infty} \frac{2^{k-1}(n-1)(n-2) \cdots
		(n-k+1)}{ n(n-2) \cdots (n-2k+4)} =  2^{k-1}\,,\]
	we can bound the total mistakes made by the algorithm in terms of \OPT:
	\[ \frac{ALG}{OPT} \leq \frac{OPT + d {n-1 \choose k-1}}{OPT} \leq 1 + \frac{d {n-1 \choose k-1}} {2d {n/2 \choose k-1}/k }  = 1 +
	\frac{k}{2}\frac{{n-1 \choose k-1}}{{n/2 \choose k-1}} \longrightarrow 1 + k\, 2^{k-2}\,, \] 
	as $n \rightarrow \infty$.
	We finally note that when $k = 3$, the result holds for all~$n$, not just in the limit, since 
	\[ \frac{2^{k-1}(n-1)(n-2) \cdots (n-k+1)}{ n(n-2) \cdots (n-2k+4)} = \frac{2^2 (n-1)(n-2)}{n(n-2)}  \leq 4\,. \]
	A similar argument applies when~$k = 2$, which in essence is what allowed Bansal et al.~\cite{Bansal2004correlation}
	to develop a $3$-approximation for this case, independent of~$n$.
\end{proof}
Although the exponential dependence on~$k$ in makes this result a poor approximation for large motifs, at least in the case $k = 3$, this is a~$7$-approximation for all~$n$, not just for large~$n$.

\section{Discussion}
We have demonstrated a~$\Theta(\log n)$ integrality gap for the \lcc LP relaxation, which highlights why previous attempts to obtain a constant-factor
approximation via LP rounding have failed.
It remains an open question whether better approximation factors exist for small values of~$\lambda$ in
$O(1/\log n)$. For minimizing disagreements, there are relatively few techniques that don't rely on the LP relaxation that lead to approximations better than $O(\log n)$ for different variants of correlation clustering. The next step is either to develop an entirely new approach or prove further hardness results for approximating \lcc when $\lambda$ is small.

For \mcc, we have given an approximation algorithm for arbitrary (constant) hyperedge size~$k$ that is linear in~$k$,
and provided a first combinatorial approximation result, which avoids solving an LP relaxation, for to the two-cluster case.
An interesting open question is whether a pivoting algorithm \`a la Ailon et al.~\cite{ailon2008aggregating} could be developed for the \mcc objective. For maximizing agreements, the simple strategy of either placing all nodes together or separating all nodes into singletons will still lead to a 1/2-approximation for hypergraphs with arbitrary weights and any $k$. This leads to open questions about what results for maximizing agreements can be generalized to the hypergraph setting. Another open question is whether an approximation that is independent of $k$ could be developed for minimizing disagreements in hypergraphs.
\bibliography{gvw-isaac}
\bibliographystyle{plain}

\newpage
\appendix

\section{Approximation Results for \lcc when $\lambda < 1/2$}
A pivoting algorithm for \cc operates by repeatedly selecting an unclustered node from the graph,
and assigning it to a cluster with all of its (positive) neighbors that have yet to be clustered. The algorithm repeats this procedure until all nodes are clustered. Ailon et al.\ gave an approximation result for this method when nodes are chosen uniformly at random~\cite{ailon2008aggregating}; van Zuylen and Williamson later developed deterministic pivoting algorithms based on a careful selection of pivot nodes~\cite{zuylen2009deterministic}.

In previous work we gave a $3$-approximation for \lcc when $\lambda > 1/2$, by applying a theorem of van Zuylen and Williamson for deterministic pivoting
algorithms for \cc.
We restate a slight variant of the theorem here that is sufficient for our purposes. A full proof, including how to deterministically select pivot nodes, is given in the original work of van Zuylen and Williamson~\cite{zuylen2009deterministic}.
\begin{theorem}
	\label{thm:31}
	(\cite[Theorem~3.1]{zuylen2009deterministic})
	Consider an instance of weighted \cc, $G = (V, (w_{ij}^+),(w_{ij}^-))$, a set of associated LP costs $\{c_{ij} : i \in V, j \in V, i \neq j \}$ and another graph $G' = (V,
	F^+,F^-)$, where $\{ F^+, F^-\}$ partitions all pairs of nodes $V \times V$ in such a way that
	\begin{itemize}
		\item $w_{ij}^- \leq \alpha c_{ij}$ for all $(i,j) \in F^+$ and $w_{i,j}^+ \leq \alpha c_{ij}$ for all $(i,j) \in F^-$,
		\item $w_{ij}^+ + w_{jk}^+ + w_{ik}^- \leq \alpha(c_{ij} + c_{jk} + c_{ik})$ for every \emph{bad triplet}: $(i,j) \in F^+, (j,k) \in F^+$ and $(i,k) \in F^-$.
	\end{itemize}
	There exists a deterministic pivoting algorithm which, when applied to~$G'$, produces an output within a factor~$\alpha$ of the optimum
	for~$G$.
\end{theorem}
Pseudocode for our algorithm, \textsc{ThreeLP}, shown to be a factor-$3$ approximation~\cite{veldt2018correlation}, is given in Algorithm~\ref{alg:LP3}.
\begin{algorithm}[tb]
	\caption{\textsc{ThreeLP}}
	\begin{algorithmic}[5]
		\State{\bfseries Input:} An instance of \lcc: $G = (V,E^+,E^-)$, $\lambda  \in (0,1)$
		\State Solve the \lcc LP relaxation~\eqref{eq:lcc}.
		\State Define $G' = (V,F^+,F^-)$ where
		\[ F^+ = \{(i,j) : x_{ij} < 1/3 \}, \hspace{.5cm} F^- = \{(i,j) : x_{ij} \geq 1/3 \}  \]
		\State Apply a (randomized or deterministic) pivoting algorithm on $G'$.
	\end{algorithmic}
	\label{alg:LP3}
\end{algorithm}
We extend the approximation guarantees for \threeLP, to include the values of $\lambda < 1/2$.
We also note that a similar algorithm produces a 2-approximation for cluster deletion (i.e. when $\lambda \geq |E|/(1+|E|)$)~\cite{veldt2018correlation}, but this relies on a slightly different construction of the new signed graph $G' = (V,F^+,F^-)$. Therefore, we just focus on the approximation guarantees that hold for Algorithm~\ref{alg:LP3}.
\begin{theorem} (Theorem \ref{thm:lamcc} in main text)
	Algorithm~\ref{alg:LP3} returns a $3$-approximation for \lcc when $\lambda \geq 1/2$. When $\lambda < 1/2$, it returns an~$\alpha$ approximation,
	where $\alpha = \max \left \{\,\frac{1}{\lambda}, \frac{6 -3\lambda}{1 +\lambda} \right \}$. 	
\end{theorem}
\begin{proof}
	We show that the assumptions of Theorem~\ref{thm:31} hold for the specific approximation factors.
	Many aspects of the full proof for the~$\lambda \geq 1/2$ case~\cite{veldt2017lamcc} directly apply here, regardless of the value of~$\lambda$.
	In particular, the inequalities $w_{ij}^- \leq \alpha c_{ij}$ for all $(i,j) \in F^+$ and $w_{i,j}^+ \leq \alpha c_{ij}$ for all $(i,j) \in F^-$ hold independent
	of~$\lambda$.
	Next, we consider the second inequality
	\begin{equation}
	\label{ineq} 
	w_{ij}^+ + w_{jk}^+ + w_{ik}^- \leq \alpha(c_{ij} + c_{jk} + c_{ik})\,,
	\end{equation}
	which must hold for every triplet $\{i,j,k\}$ such that $(i,j) \in F^+, (j,k) \in F^+$ and $(i,k) \in F^-$.
	To show this, we must consider all possible types of edges that could be shared by nodes $\{i,j,k\}$ in the original graph $G =
	(V,E^+,E^-)$. We look at three cases that are central for understanding the approximation guarantees when~$\lambda < 1/2$. Recall from Algorithm~\ref{alg:LP3} that if $x_{uv} \geq 1/3$ we make $(u,v)$ a negative edge in $G'$, and otherwise we make it a positive edge. Therefore, if $\{i,j,k\}$ is a bad triangle in~$G'$ in which $(i,k) \in F^-$ is the negative edge, then $x_{ij} \leq 1/3$, $x_{jk} \leq 1/3$, and $x_{ik} > 1/3$.
	
	\paragraph{Case 1: $(i,j) \in E^+, (j,k) \in E^+$ and $(i,k) \in E^-$}
	Given these types of edges in the original graph, we know that $c_{ij} = (1-\lambda)x_{ij}$, $c_{jk} = (1-\lambda)x_{ik}$,
	and $c_{ik} = \lambda(1-x_{ik})$.
	Therefore,
	\begin{align*}
	\alpha (c_{ij} &+ c_{jk} + c_{ik}) = \alpha \left((1-\lambda)(x_{ij} + x_{jk}) + \lambda(1-x_{ik} )  \right) \\
	&\geq \alpha \left( (1-\lambda)x_{ik} + \lambda(1-x_{ik}) \right) = \alpha \left( (1-2 \lambda) x_{ik} + \lambda \right) \\
	&> \alpha\left( (1-2\lambda) 1/3 + \lambda \right) = \alpha (1+\lambda)/{3}\,.
	\end{align*}
	The weights satisfy $w_{ij}^+ = w_{jk}^+ = 1- \lambda$ and $w_{ik}^- = \lambda$.
	With a few steps of algebra we can see that the above expression is an upper bound for $2-\lambda = w_{ij}^+ + w_{jk}^+ + w_{ik}^-$ (the right hand side of inequality~\eqref{ineq}) as long as $\alpha \geq \frac{6-3\lambda}{1+\lambda}$.
	
	\paragraph{Case 2: $(i,j) \in E^+, (j,k) \in E^-$ and $(i,k) \in E^-$}
	In this case, the LP costs are $(c_{ij},c_{jk},c_{ik}) = ( (1-\lambda)x_{ij}, \lambda (1-x_{jk}), \lambda (1-x_{ik})
	)$ and the weights are $(w_{ij}^+,w_{jk}^+,w_{ik}^-) = (1-\lambda, 0, \lambda)$. Therefore,
	\begin{align*}
	\alpha (c_{ij} & + c_{jk} + c_{ik}) = \alpha \left( (1-\lambda) x_{ij} +\lambda(1-x_{jk}) + \lambda(1-x_{ik})\right)\\
	& \geq \alpha \left( \lambda -\lambda x_{jk} + \lambda - \lambda x_{ik} \right) \geq \alpha \left( \lambda - \lambda/3 + \lambda - 2\lambda/3 \right)\\
	&= \alpha\lambda \geq 1 = w_{ij}^+ + w_{jk}^+ + w_{ik}^-
	\end{align*}
	which holds as long as~$\alpha \geq 1/\lambda$.
	\paragraph{Case 3: $(i,j) \in E^-, (j,k) \in E^+$ and $(i,k) \in E^-$}
	This case is symmetric to Case~2. The proof follows by simply switching the roles of edges $(i,j)$ and $(j,k)$.
	
	A full proof of the remaining cases, which all hold independent of $\lambda$, is given in previous work~\cite{veldt2017lamcc}. We
	therefore see that if $\alpha = \max \left \{\, \frac{1}{\lambda}, \frac{6- 3 \lambda}{1 +\lambda} \right \}$, the full result holds.
\end{proof}

By solving ${1}/\lambda = (6- 3 \lambda)/(1 +\lambda)$ for $\lambda$, we find that the behavior of the approximation factor changes when $\lambda = (5-
\sqrt{13})/6 \approx 0.2324$. For~$\lambda$ greater than this threshold, the approximation factor is always between~$3$ and~$4.303$.  

%

\end{document}